\documentclass[12pt,a4wide]{article}

\usepackage{amsthm,amssymb,mathrsfs,setspace,booktabs}
\usepackage{mathtools,amsmath,nccmath}

\usepackage{pstricks}
\usepackage{array}
\newcolumntype{P}[1]{>{\centering\arraybackslash}p{#1}}

\usepackage{tikz}

\usepackage{subcaption}

\usepackage{hyperref}

\usepackage{appendix}
\usepackage{caption}
\usepackage[]{biblatex} % we will use biblatex + refs.bib
\usepackage{fancyhdr}
\usepackage{multicol}
\usepackage{geometry}
\usepackage{xcolor}
\hypersetup{
    colorlinks,
    linkcolor={blue!50!black},
    citecolor={blue!50!black},
    urlcolor={blue!80!black}
}

\usepackage{graphicx}
\usepackage[nottoc]{tocbibind}

\usepackage{braket}
\usepackage[nameinlink,capitalize]{cleveref}
\usepackage{bbm}

\theoremstyle{plain}
\newtheorem{theorem}{Theorem}[section]

\newtheorem{proposition}[theorem]{Proposition}

\theoremstyle{definition}

\theoremstyle{remark}
\newtheorem{remark}[theorem]{Remark}

\renewcommand{\today}{\ifcase \month \or January\or February\or March\or April\or May%
\or June\or July\or August\or September\or October\or November\or December\fi\ %
\number \year} 

\newcommand{\Tr}{\operatorname{Tr}}

\begin{document}

%%%%%%%%%%%%%%%%%%%%%%%%%%%%%%%%%%%%%%%%%%%%%%%%%%%%%%%%%%%%%%%%%%%%%%%%%%%%%%%%%%%%%%%%%%%%%%%%
% Intro pages (title, certificate, acknowledgements, abstract, etc.)
%%%%%%%%%%%%%%%%%%%%%%%%%%%%%%%%%%%%%%%%%%%%%%%%%%%%%%%%%%%%%%%%%%%%%%%%%%%%%%%%%%%%%%%%%%%%%%%%

%----------------------------------------------------------------------------------------
%	TITLE PAGE
%----------------------------------------------------------------------------------------
\begin{titlepage}
\enlargethispage{3cm}

\begin{center}

\vspace*{-1cm}

\textbf{\Large Qudits offer no advantages over dits for sending random messages}\\[10pt]

\vspace*{0.5cm}

% ALTERNATIVE COVER PAGES:
% Uncomment lines 27-33 and hide lines 19-25 for *MINOR PROJECT COVER PAGE*
% Uncomment lines 35-41 and hide lines 19-25 for *PhD THESIS COVER PAGE*

A Thesis Submitted \\
to the University of Texas at Austin  \\
for the Turing Scholars program \\

% A Minor Project Report Submitted in Partial Fulfillment of the Requirements for \\
%% \vspace{0.5cm}
% {\Large \bf MINOR DEGREE}\\
% \vspace{0.3cm}
% in\\ 
% \vspace{0.3cm}
% {\large \bf [Department Name] } \\

% A thesis submitted for the degree of\\
%% \vspace{0.5cm}
% {\Large \bf DOCTOR OF PHILOSOPHY}\\
% \vspace{0.3cm}
% in\\ 
% \vspace{0.3cm}
% {\large \bf [Department Name] } \\

                      \vspace{10mm}
                   {\em  by} \\ \vspace{3mm}
             {\large \bf Ronit Amit Shah}\\[.3in]

\begin{figure}[h]
  \begin{center}
\includegraphics[height=50mm]{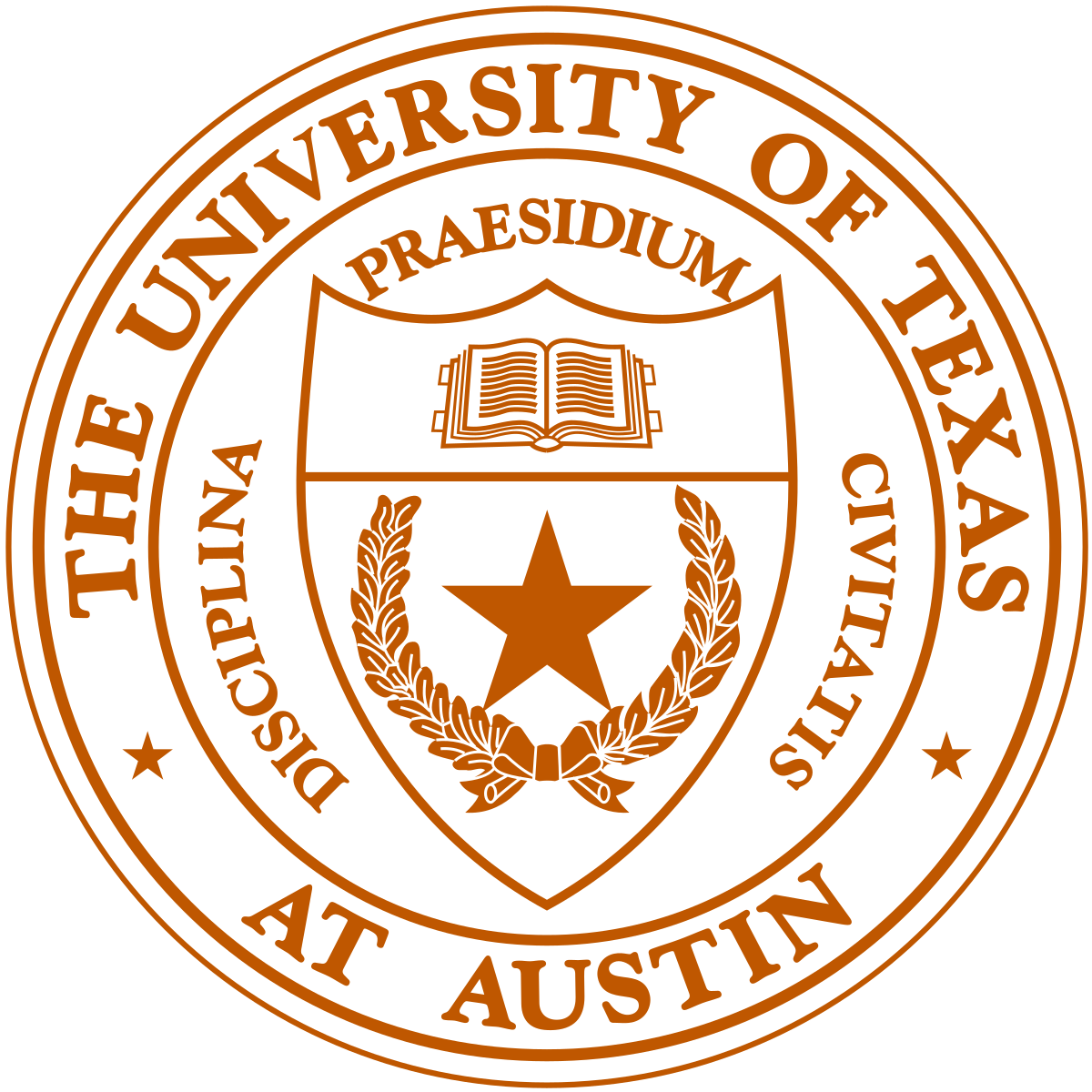}
  \end{center}
\end{figure}
\vspace*{0.2cm}

{\em\large to }\\%[8pt]
{\bf\large Department of Computer Science} \\%[4pt] Uncomment if not applicable (Humanities/Data Science)
{\bf\large College of Natural Sciences}\\%[4pt]
{\bf\large University of Texas at Austin}\\%[8pt]
{\it\large December 2025}

\end{center}

\end{titlepage}

\clearpage

% If you have a certificate page, include it here
% \input{00_Intro_Pages/03b_Certificate}

%----------------------------------------------------------------------------------------
%	ACKNOWLEDGMENTS
%----------------------------------------------------------------------------------------
\begin{center}
{\large{\bf{ACKNOWLEDGEMENT}}}
\end{center}

I would first like to thank my supervisor, Scott Aaronson, for his constant support and guidance throughout this project. His patience in teaching me quantum mechanics, his intuition in pointing me toward the relevant literature and authors, and his careful feedback on multiple drafts have all been invaluable in shaping this work.

I am also sincerely grateful to my committee members, Nick Hunter-Jones and William Kretschmer, whose flexibility and responsiveness made it possible to carry out this project smoothly. Their time, comments, and encouragement played an important role in bringing this thesis to completion.

Finally, I want to express my deepest appreciation to my friends and family for their constant encouragement and support at every step. Their belief in me and their care outside of academics have been essential to seeing this work through.

\clearpage
%----------------------------------------------------------------------------------------
%	ABSTRACT
%----------------------------------------------------------------------------------------
\vspace{6pt}
\begin{flushleft}
    \setlength{\parskip}{0pt}
        %\bigskip
    {\centering{{\Large{\bf{ABSTRACT}}}} \par}
    \bigskip
	\hrule \vspace{1.5cm} % Horizontal line	
\end{flushleft} % This section is not essential for the abstract

We consider the following simple scenario: Alice has one of many possible messages, drawn from a known distribution, and wants to maximize the probability that Bob guesses her message correctly.  We prove that if Alice can send only a qudit to Bob, without preshared entanglement, there is never any advantage over sending him a classical dit. This result was previously known only for a uniform distribution.

We also prove a mixed-state generalization of this result in the form of an upper bound on the success probability of discriminating between mixed quantum states with a single measurement. This bound is based solely on the dimension, probability distribution, and eigenvalues of the states and is sharp among such bounds.  
\vspace{3cm} % Reduce if text overflowing to a new page. Don't make it too long.

\clearpage

\pagenumbering{arabic}
\setcounter{page}{1}

% ========================== Main content starts here ========================== %

\section{Introduction}

\subsection*{Classical messenger--receiver task}
Consider the following classical scenario. Alice wishes to send Bob one of $m$ messages $i\in\{1,\dots,m\}$ drawn from a known prior distribution $(p_i)$, and Bob must output a guess $\hat{\imath}$ for which message was sent. Alice is allowed to send only a single classical symbol from an alphabet of size $d$ (a \emph{dit}), and Bob can use any decoding rule from symbols to guesses. The goal is to maximize the success probability
\[
P_{\mathrm{succ}}^{\mathrm{cl}}=\Pr[\hat{\imath}=i].
\]

No matter how Alice and Bob choose their encoding and decoding, this classical success probability is upper-bounded by the sum of the $d$ largest priors:
\[
P_{\mathrm{succ}}^{\mathrm{cl}}\ \le\ \sum_{k=1}^{d} p_{(k)},
\]
where $p_{(1)}\ge p_{(2)}\ge\cdots\ge p_{(m)}$ denotes the priors ordered from largest to smallest. Intuitively, each of the $d$ possible signals can reliably single out at most one message; the best one can do is dedicate the signals to the $d$ most likely messages. This can be formalized using rearrangement/majorization arguments and the comparison of statistical experiments~\cite[Ch.~10]{HardyLittlewoodPolya,Blackwell1953,Torgersen1991}. The bound is tight: by assigning distinct signals to the $d$ most probable messages and ignoring the rest, one achieves $P_{\mathrm{succ}}^{\mathrm{cl}}=\sum_{k=1}^{d} p_{(k)}$.

\subsection*{Quantum version: sending one qudit}
Now consider a quantum version of the same task. Alice again has messages $i\in\{1,\dots,m\}$ drawn from the same priors $(p_i)$, but now she is allowed to send a single \emph{qudit} to Bob, i.e., a quantum state $\rho_i$ of some fixed dimension $d$. Bob is free to attach any local ancilla, apply any quantum circuit, and perform any projective measurement in a large space, followed by a classical decision rule. In particular, Bob’s measurement can have \emph{more} than $d$ possible outcomes because of the ancilla---so it is not a priori obvious that the classical $d$-signal ceiling still applies in this setting.

The overall success probability is
\[
P_{\mathrm{succ}}^{\mathrm{qu}}=\Pr[\hat{\imath}=i],
\]
where $\hat{\imath}$ is Bob’s guess. This is exactly the minimum-error quantum state discrimination problem for the ensemble $\{(p_i,\rho_i)\}$~\cite{Helstrom1976,BarnettCroke2009,Watrous2018}. A priori Alice might choose mixed states $\rho_i$, but the success probability is affine (linear) in each $\rho_i$ for fixed measurement, and the set of density operators is convex. Therefore, for fixed priors and dimension, the optimal success probability is achieved by an ensemble of \emph{pure} states. Mixed states are never needed to improve the optimum for this messenger--receiver task.

We can also harmlessly compress the Hilbert space: if the states $\{\rho_i\}$ jointly lie in a subspace of dimension $d'<d$, we can project onto that subspace and treat its dimension as the effective $d$. Throughout, when we say ``dimension $d$ of the quantum states'' we mean the dimension of this joint support.

By contrast, in the presence of preshared entanglement between Alice and Bob, protocols such as superdense coding show that a single transmitted qubit combined with one entangled pair can convey the information of two classical bits. Our result complements this: in the \emph{absence} of preshared entanglement, a single qudit has no one-shot advantage over a classical $d$-ary symbol for identifying a randomly chosen message.

\subsection*{Main results}
Our first main result shows that, even with arbitrary non-uniform priors and arbitrary pure-state encodings on a $d$-dimensional system, a qudit without preshared entanglement cannot outperform a classical $d$-ary signal for one-shot message identification.

\begin{theorem}[Qudits vs.\ dits for pure encodings]\label{thm:pure-main}
Let Alice’s messages have priors $(p_i)_{i=1}^m$, and let these priors be sorted as $p_{(1)}\ge\cdots\ge p_{(m)}$. Suppose that for each $i$ Alice encodes message $i$ into a pure state $\ket{\psi_i}$, and that the span of the states $\{\ket{\psi_i}\}$ has dimension $d$. Then for any one-shot measurement and decision rule Bob may apply,
\[
P_{\mathrm{succ}}^{\mathrm{qu}}\ \le\ \sum_{k=1}^{d}p_{(k)}.
\]
This bound is tight: if the $d$ most likely states are mutually orthogonal, Bob can measure in a basis extending them and achieve success probability $\sum_{k=1}^{d}p_{(k)}$.
\end{theorem}

Because the optimal communicator never needs mixed states, \cref{thm:pure-main} already suffices to answer the original messenger--receiver question: a single qudit without preshared entanglement is no more powerful than a classical dit for one-shot message identification.

Our second main result is a mixed-state generalization stated in terms of the eigenvalues of the weighted states $p_i\rho_i$.

\begin{theorem}[Mixed-state spectral bound]\label{thm:mixed-main}
Let Alice’s messages have priors $(p_i)$ and be encoded into mixed states $\rho_i$ of dimension $d$. Let $\lambda_{ik}$ denote the eigenvalues of $\rho_i$, and define $\lambda'_{ik}=p_i\lambda_{ik}$. Form the multiset $\{\lambda'_{ik}\}_{i,k}$ and sort it as
\[
\lambda'_{(1)}\ \ge\ \lambda'_{(2)}\ \ge\ \cdots.
\]
Then for any one-shot measurement and decision rule,
\[
P_{\mathrm{succ}}^{\mathrm{qu}}\ \le\ \sum_{k=1}^{d}\lambda'_{(k)}.
\]
This bound is tight given only the multiset of weighted eigenvalues and the dimension $d$: there exist ensembles and measurements that achieve equality.
\end{theorem}

While \cref{thm:mixed-main} is not needed to establish the qudit--vs.--dit statement (because optimal encodings can be taken pure), it is useful in its own right. Since our measurement model is equivalent in power to arbitrary POVMs and post-processing~\cite{Helstrom1976,Watrous2018}, \cref{thm:mixed-main} gives a general spectral upper bound on the optimal success probability of minimum-error discrimination for any finite ensemble of mixed states, in terms only of their dimension, priors, and weighted eigenvalues.

Both bounds strictly strengthen the standard dimension-only ceiling $P_{\mathrm{succ}}^{\mathrm{qu}}\le d\,p_{\max}$ and mirror the classical messenger--receiver limit based on $d$ possible signals.

\subsection*{Illustrative example: one bit/qubit and three messages}
As a concrete illustration, suppose Alice wishes to send one of three messages to Bob with priors
\[
(p_1,p_2,p_3)=\Big(\tfrac{1}{2},\tfrac{1}{3},\tfrac{1}{6}\Big).
\]

\paragraph{Classical case.}
If Alice is allowed to send only a single classical bit (two possible signals, so $d=2$), the classical bound yields
\[
P_{\mathrm{succ}}^{\mathrm{cl}}\ \le\ \sum_{k=1}^{2}p_{(k)}\ =\ \tfrac{1}{2}+\tfrac{1}{3}\ =\ \tfrac{5}{6},
\]
and this value is achievable by dedicating the two bit values to the two most likely messages and never explicitly encoding the third.

\paragraph{Quantum case.}
Now let Alice be allowed to send a single qubit (again $d=2$). By \cref{thm:pure-main} (or \cref{thm:mixed-main} together with the fact that pure states suffice),
\[
P_{\mathrm{succ}}^{\mathrm{qu}}\ \le\ \sum_{k=1}^{2}p_{(k)}\ =\ \tfrac{1}{2}+\tfrac{1}{3}\ =\ \tfrac{5}{6}.
\]
This upper bound is tight: Alice can encode the two most likely messages as orthogonal qubit states and sacrifice the least likely one; Bob measures in the corresponding basis and decodes accordingly. Thus for this example the optimal quantum and classical success probabilities \emph{coincide}, and our general theorems show this remains true for any priors and any dimension $d$.

\paragraph{Structure.}
In \cref{sec:related} we briefly review related classical and quantum results and compare our bounds to existing ones, including for the preceding example. \Cref{sec:measurement} introduces the measurement model we use. \Cref{sec:pure-proof} proves \cref{thm:pure-main} using that model and explains why the bound is sharp. \Cref{sec:mixed} proves the mixed-state spectral bound and its sharpness. We conclude in \cref{sec:conclusion}.
\section{Related work}\label{sec:related}
The classical messenger--receiver task, and more generally the comparison of statistical experiments, is treated in detail in the work of Blackwell and others~\cite{Blackwell1953,Torgersen1991}. The fact that with only $d$ possible signals the optimal Bayes success probability is at most $\sum_{k=1}^{d}p_{(k)}$ follows from standard majorization inequalities~\cite[Ch.~10]{HardyLittlewoodPolya} and can be phrased in Blackwell’s order as saying that an experiment which perfectly reveals one of the top-$d$ labels is maximal for $0$--$1$ loss.

On the quantum side, minimum-error state discrimination has been studied since the early days of quantum detection theory~\cite{Helstrom1976,YuenKennedyLax1975}. For $m=2$ states, the exact optimum is given by Helstrom’s trace-norm formula~\cite{Helstrom1976}. For larger $m$, no closed form is known in general, but many useful bounds and constructions exist; see \cite{BarnettCroke2009,Watrous2018} for overviews.

There are two broad families of one-shot upper bounds on $P_{\mathrm{succ}}^{\mathrm{qu}}$:

\paragraph{Dimension-only quantum ceilings.}
A standard support-dimension argument (or a simple feasible dual point) yields
\[
P_{\mathrm{succ}}^{\mathrm{qu}}\ \le\ d\,\max_i p_i\|\rho_i\|_\infty,
\]
which for pure states simplifies to $P_{\mathrm{succ}}^{\mathrm{qu}}\le d\,p_{\max}$ and, under equal priors, to $P_{\mathrm{succ}}^{\mathrm{qu}}\le d/m$; see, e.g., \cite[Ch.~3]{Watrous2018} and \cite{BarnettCroke2009}. This coarse ceiling is tight on many symmetric/equiprobable ensembles (trine, tetrahedral, SICs).

For the illustrative example in the introduction with priors $(1/2,1/3,1/6)$ and $d=2$, this bound gives
\[
P_{\mathrm{succ}}^{\mathrm{qu}}\ \le\ d\,p_{\max} = 2\cdot\tfrac{1}{2} = 1,
\]
which is vacuous. In contrast, our bound \cref{thm:pure-main} yields the exact optimum $P_{\mathrm{succ}}^{\mathrm{qu}}=5/6$ using only the priors and the dimension. More generally, the bounds we prove in \cref{thm:pure-main,thm:mixed-main} use only the same coarse information (dimension plus priors/eigenvalues) but strictly strengthen this ceiling: we replace $d\,p_{\max}$ by $\sum_{k=1}^{d}p_{(k)}$ for pure states and by $\sum_{k=1}^{d}\lambda'_{(k)}$ for mixed states.

\paragraph{Overlap-/geometry–sensitive bounds.}
Many sharper bounds incorporate detailed information about the overlaps or geometry of the states. These include near-optimality guarantees for the pretty-good measurement~\cite[Ch.~3]{Watrous2018}, computable dual bounds with attainability conditions~\cite{Nakahira2018}, and recent analytical bounds that recover the Helstrom formula in the binary case~\cite{Loubenets2022,PRATight2024}. Such bounds can be extremely tight instance by instance but require pairwise overlaps or explicit dual operators.

In contrast, our bounds depend only on the priors $(p_i)$, the eigenvalues of $\{p_i\rho_i\}$, and the dimension $d$; they are sharp given that information and match the classical $d$-signal ceiling in the pure case.
\section{Measurement model}\label{sec:measurement}

We now formalize the one-shot quantum measurement in a way that stays close to first principles and makes the role of the dimension $d$ explicit.

\paragraph{Compressing to the joint support.}
Let $\{\rho_i\}$ be Alice’s ensemble of states. Let $\mathcal{H}$ be the span of the supports of the $\rho_i$, and let $d=\dim(\mathcal{H})$. Because any measurement and decision rule can be preceded by a projection onto $\mathcal{H}$ without loss of information, we can work entirely in this $d$-dimensional space.

\paragraph{Measurement and post-processing.}
Any one-shot quantum measurement can be realized by:
\begin{itemize}
  \item attaching an ancilla in some fixed state,
  \item applying a unitary on system plus ancilla, and
  \item measuring in a fixed orthonormal basis and then classically post-processing the outcome.
\end{itemize}
We model this by choosing an arbitrary $q\times d$ matrix $V$ with orthonormal columns,
\[
V^\dagger V = I_d,
\]
and measuring in the standard basis $\{\ket{j}\}_{j=1}^q$ of the $q$-dimensional space. For an input pure state $\ket{\psi}$, this produces outcome $j\in\{1,\dots,q\}$ with probability
\[
\Pr[j|\psi]=\bigl|\braket{j|V|\psi}\bigr|^2.
\]
For a mixed state $\rho$, the outcome probabilities are $\Pr[j|\rho]=\bra{j}V\rho V^\dagger\ket{j}$.

After observing $j$, Bob applies a deterministic decision rule
\[
g:\{1,\dots,q\}\to\{1,\dots,m\}
\]
and outputs $g(j)$ as his guess.

\paragraph{Success probability.}
For a pure-state ensemble $\{(p_i,\ket{\psi_i})\}$, the overall success probability of a given $V$ and $g$ is
\begin{equation}\label{eq:Ps-pure-model}
P_{\mathrm{succ}}^{\mathrm{qu}}
=\sum_{j=1}^q p_{g(j)}\,\bigl|\braket{j|V|\psi_{g(j)}}\bigr|^2
=\sum_{i=1}^m p_i \sum_{j\in g^{-1}(i)} \bigl|\braket{j|V|\psi_i}\bigr|^2.
\end{equation}
For a mixed-state ensemble $\{(p_i,\rho_i)\}$, the success probability is
\begin{equation}\label{eq:Ps-mixed-model}
P_{\mathrm{succ}}^{\mathrm{qu}}
=\sum_{j=1}^q p_{g(j)}\,\bra{j}V\rho_{g(j)}V^\dagger\ket{j}
=\sum_{i=1}^m p_i \sum_{j\in g^{-1}(i)} \bra{j}V\rho_i V^\dagger\ket{j}.
\end{equation}

This model is equivalent in power to describing measurements by POVMs and classical post-processing (via standard dilation and coarse-graining arguments)~\cite{Helstrom1976,Watrous2018}, but we avoid POVM notation in what follows.
\section{Pure-state bound}\label{sec:pure-proof}

We now prove the pure-state qudit vs.\ dit bound using the measurement model of \cref{sec:measurement}. The argument mirrors the classical case, but with a nontrivial constraint coming from the dimension $d$.

Recall from \eqref{eq:Ps-pure-model} that for a pure-state ensemble $\{(p_i,\ket{\psi_i})\}$,
\[
P_{\mathrm{succ}}^{\mathrm{qu}}
=\sum_{i=1}^m p_i \sum_{j\in g^{-1}(i)} \bigl|\braket{j|V|\psi_i}\bigr|^2.
\]
For each message $i$, define
\[
S_i\ \coloneqq\ \sum_{j\in g^{-1}(i)} \bigl|\braket{j|V|\psi_i}\bigr|^2.
\]
Then
\[
P_{\mathrm{succ}}^{\mathrm{qu}}=\sum_{i=1}^m p_i S_i.
\]

Two simple constraints hold for the numbers $S_i$:

\paragraph{Bounded by $1$.}
Fix $i$. Then
\begin{align*}
S_i
&= \sum_{j\in g^{-1}(i)} \bigl|\braket{j|V|\psi_i}\bigr|^2
\le \sum_{j=1}^q \bigl|\braket{j|V|\psi_i}\bigr|^2 \\
&= \|V\ket{\psi_i}\|^2
 = \bra{\psi_i}V^\dagger V\ket{\psi_i}
 = \braket{\psi_i|\psi_i}
 = 1.
\end{align*}
So $0\le S_i\le 1$.

\paragraph{Total budget at most $d$.}
For any fixed outcome $j$ and state index $i$, the Cauchy--Schwarz inequality gives
\[
\bigl|\braket{j|V|\psi_i}\bigr|^2
\le \|\bra{j}V\|^2\,\|\ket{\psi_i}\|^2
= \|\bra{j}V\|^2.
\]
Thus
\[
S_i = \sum_{j\in g^{-1}(i)} \bigl|\braket{j|V|\psi_i}\bigr|^2
\le \sum_{j\in g^{-1}(i)} \|\bra{j}V\|^2.
\]
Summing over $i$,
\[
\sum_{i=1}^m S_i
\le \sum_{i=1}^m \sum_{j\in g^{-1}(i)} \|\bra{j}V\|^2
= \sum_{j=1}^q \|\bra{j}V\|^2
= \Tr(VV^\dagger)
= \Tr(V^\dagger V)
= \Tr(I_d)
= d.
\]

\paragraph{Linear-programming viewpoint.}
For any measurement $V$ and decision rule $g$, we therefore have
\[
P_{\mathrm{succ}}^{\mathrm{qu}} = \sum_{i=1}^m p_i S_i,
\]
for some numbers $(S_i)$ satisfying
\[
0\ \le\ S_i\ \le\ 1 \quad\text{for all }i,\qquad
\sum_{i=1}^m S_i\ \le\ d.
\]

To obtain a universal upper bound, we can maximize $\sum_i p_i S_i$ over all $(S_i)$ satisfying these inequalities. This is a simple linear program:
\[
\max\left\{\sum_{i=1}^m p_i S_i : 0\le S_i\le 1,\ \sum_{i=1}^m S_i\le d\right\}.
\]
Clearly the optimum is achieved by taking $S_i\in\{0,1\}$ and saturating $\sum_i S_i=d$. Thus the best we can do is to set $S_i=1$ for $d$ indices $i$ and $0$ for the rest. Ordering the priors $p_{(1)}\ge\cdots\ge p_{(m)}$, the maximum is
\[
\sum_{i=1}^m p_i S_i \le \sum_{k=1}^{d} p_{(k)}.
\]
Since the same upper bound applies for any measurement and decision rule, this proves
\[
P_{\mathrm{succ}}^{\mathrm{qu}} \le \sum_{k=1}^{d} p_{(k)},
\]
i.e., \cref{thm:pure-main}.

\paragraph{Sharpness via classical construction.}
The sharpness of \cref{thm:pure-main} follows directly from the classical tightness construction discussed earlier. Classically, assigning $d$ distinct signals to the $d$ most likely messages attains success probability $\sum_{k=1}^{d}p_{(k)}$. In the quantum setting, we can realize these $d$ signals as $d$ orthogonal pure states in dimension $d$, and any remaining messages as arbitrary superpositions in their span. Measuring in the corresponding orthonormal basis and decoding as in the classical strategy achieves the same success probability $\sum_{k=1}^{d}p_{(k)}$, so the bound is tight.
\section{Mixed-state spectral bound and sharpness}\label{sec:mixed}

We now prove the mixed-state spectral bound \cref{thm:mixed-main}. Although mixed states are not needed to optimize the messenger--receiver task, the bound is natural and may be useful in other settings where mixed-state ensembles arise, especially given its POVM-equivalent formulation.

\subsection{Setup and factorization}
Let Alice’s messages be encoded into mixed states $\rho_i$ of dimension $d$, with priors $(p_i)$. As in \eqref{eq:Ps-mixed-model},
\[
P_{\mathrm{succ}}^{\mathrm{qu}}
=\sum_{i=1}^m p_i \sum_{j\in g^{-1}(i)} \bra{j}V\rho_i V^\dagger\ket{j}.
\]

For each $i$, choose an orthogonal factorization
\[
\rho_i = \psi_i\psi_i^\dagger,
\]
where $\psi_i$ is a $d\times d$ matrix such that $\psi_i^\dagger\psi_i$ is diagonal with entries
\[
\bra{k}\psi_i^\dagger\psi_i\ket{k}=\lambda_{ik},
\]
the eigenvalues of $\rho_i$. Let $\lambda'_{ik}=p_i\lambda_{ik}$ be the eigenvalues of the weighted states $p_i\rho_i$.

With this factorization,
\[
\bra{j}V\rho_i V^\dagger\ket{j}
=\|\bra{j}V\psi_i\|^2
=\sum_{k=1}^{d} \bigl|\bra{j}V\psi_i\ket{k}\bigr|^2,
\]
so the success probability is
\begin{equation}\label{eq:Ps-mixed-factor}
P_{\mathrm{succ}}^{\mathrm{qu}}
=\sum_{i=1}^m p_i\sum_{j\in g^{-1}(i)}\sum_{k=1}^{d} \bigl|\bra{j}V\psi_i\ket{k}\bigr|^2.
\end{equation}

For each pair $(i,k)$ we now define
\[
S_{ik}\ \coloneqq\ \sum_{j\in g^{-1}(i)} \|\bra{v_{jk}}\|^2,
\]
where $\bra{v_{jk}}$ will denote the projection of the row vector $\bra{j}V$ onto the one-dimensional subspace spanned by $\psi_i\ket{k}$, as described below. These $S_{ik}$ play the same role as the $S_i$ in the pure-state proof.

\subsection{Bounding contributions of each eigen-component}
Fix a message index $i$ and eigen-index $k$. For an outcome $j\in g^{-1}(i)$, consider the projection of the row vector $\bra{j}V$ onto the one-dimensional subspace spanned by $\psi_i\ket{k}$. Let $\bra{v_{jk}}$ denote this projection, so
\[
\bra{j}V\psi_i\ket{k} = \bra{v_{jk}}\psi_i\ket{k}.
\]
Then
\[
\bigl|\bra{j}V\psi_i\ket{k}\bigr|^2
= \|\bra{v_{jk}}\|^2\,\|\psi_i\ket{k}\|^2
= \lambda_{ik}\,\|\bra{v_{jk}}\|^2.
\]
Summing over $j\in g^{-1}(i)$,
\begin{equation}\label{eq:component-mixed}
\sum_{j\in g^{-1}(i)} \bigl|\bra{j}V\psi_i\ket{k}\bigr|^2
= \lambda_{ik} \sum_{j\in g^{-1}(i)} \|\bra{v_{jk}}\|^2
= \lambda_{ik} S_{ik}.
\end{equation}

As before, we need two simple bounds on $S_{ik}$.

First, summing \eqref{eq:component-mixed} over all outcomes $j$,
\begin{align*}
\sum_{j=1}^q \bigl|\bra{j}V\psi_i\ket{k}\bigr|^2
&= \|V\psi_i\ket{k}\|^2
= \bra{k}\psi_i^\dagger V^\dagger V\psi_i\ket{k}
= \bra{k}\psi_i^\dagger\psi_i\ket{k}
= \lambda_{ik},
\end{align*}
so
\[
\sum_{j\in g^{-1}(i)} \bigl|\bra{j}V\psi_i\ket{k}\bigr|^2 \le \lambda_{ik}
\]
and hence, using \eqref{eq:component-mixed},
\[
S_{ik} \le 1.
\]

Second, for each fixed outcome $j$, the vectors $\{\psi_i\ket{k}\}_{k=1}^{d}$ (for fixed $i$) are orthogonal, so their spans are orthogonal subspaces. The projections $\{\bra{v_{jk}}\}_k$ of a fixed row $\bra{j}V$ onto these subspaces satisfy
\[
\sum_{k=1}^{d}\|\bra{v_{jk}}\|^2 \le \|\bra{j}V\|^2.
\]
Summing over $j$,
\begin{equation}\label{eq:vjk-budget}
\sum_{j=1}^q\sum_{k=1}^{d}\|\bra{v_{jk}}\|^2
\le \sum_{j=1}^q\|\bra{j}V\|^2
= \Tr(VV^\dagger)
= \Tr(V^\dagger V)
= \Tr(I_d) = d.
\end{equation}
Using the definition of $S_{ik}$ and summing over $i$ and $k$,
\[
\sum_{i=1}^m\sum_{k=1}^{d} S_{ik}
= \sum_{i=1}^m\sum_{k=1}^{d}\sum_{j\in g^{-1}(i)} \|\bra{v_{jk}}\|^2
= \sum_{k=1}^{d}\sum_{j=1}^q \|\bra{v_{jk}}\|^2
\le d.
\]
We also clearly have $S_{ik}\ge 0$.

\subsection{Linear-programming formulation}
Using \eqref{eq:component-mixed} in \eqref{eq:Ps-mixed-factor}, we find
\begin{align*}
P_{\mathrm{succ}}^{\mathrm{qu}}
&= \sum_{i=1}^m p_i\sum_{k=1}^{d}\sum_{j\in g^{-1}(i)} \bigl|\bra{j}V\psi_i\ket{k}\bigr|^2 \\
&= \sum_{i=1}^m p_i\sum_{k=1}^{d} \lambda_{ik} S_{ik} \\
&= \sum_{i=1}^m\sum_{k=1}^{d} \lambda'_{ik}\,S_{ik}.
\end{align*}
For any measurement $V$ and decision rule $g$, the numbers $S_{ik}$ satisfy
\[
0\ \le\ S_{ik}\ \le\ 1,\qquad \sum_{i,k} S_{ik}\ \le\ d.
\]

To obtain a universal upper bound, we again maximize the linear form
\[
\sum_{i=1}^m\sum_{k=1}^{d} \lambda'_{ik}\,S_{ik}
\]
over all such $(S_{ik})$:
\[
\max\left\{ \sum_{i=1}^m\sum_{k=1}^{d} \lambda'_{ik}\,S_{ik} : 0\le S_{ik}\le 1,\ \sum_{i,k} S_{ik}\le d \right\}.
\]
The optimum is attained by taking $S_{ik}\in\{0,1\}$ and saturating $\sum_{i,k} S_{ik}=d$. In other words, we should set $S_{ik}=1$ for those $d$ pairs $(i,k)$ corresponding to the largest values of $\lambda'_{ik}$ and $S_{ik}=0$ elsewhere. If we sort the full multiset $\{\lambda'_{ik}\}_{i,k}$ as $\lambda'_{(1)}\ge\lambda'_{(2)}\ge\cdots$, this yields
\[
P_{\mathrm{succ}}^{\mathrm{qu}}
\le \sum_{k=1}^{d} \lambda'_{(k)},
\]
proving \cref{thm:mixed-main}.

\subsection{Sharpness and construction}
As in the pure-state case, the bound is sharp with respect to the weighted eigenvalues.

\begin{proposition}[Sharpness for mixed encodings]
Fix a dimension $d$ and the multiset of weighted eigenvalues $\{\lambda'_{ik}=p_i\lambda_{ik}\}_{i,k}$. Then there exists an ensemble of mixed states $\{\rho_i\}$ with these weighted eigenvalues and a measurement/decision rule such that
\[
P_{\mathrm{succ}}^{\mathrm{qu}} = \sum_{k=1}^{d}\lambda'_{(k)}.
\]
\end{proposition}

\begin{proof}
Let $\lambda'_{(1)}\ge\cdots\ge\lambda'_{(d)}$ be the $d$ largest weighted eigenvalues, and let $\{\ket{e_1},\dots,\ket{e_d}\}$ be an orthonormal set in the $d$-dimensional space. Construct mixed states $\rho_i$ so that these $d$ eigenvalues correspond to orthogonal eigenvectors among the $\ket{e_k}$, and assign each $\lambda'_{(k)}$ to its appropriate message index $i$ according to the original labeling. Embed all remaining spectral components in the span of $\{\ket{e_k}\}$ so that the overall dimension remains $d$.

Bob measures in a basis extending $\{\ket{e_1},\dots,\ket{e_d}\}$ and, upon outcome $\ket{e_k}$, guesses the message index associated with $\lambda'_{(k)}$. If the true message is that index, the probability of outcome $\ket{e_k}$ is exactly $\lambda'_{(k)}$; other messages contribute nothing guaranteed on those outcomes. Summing over $k$ gives
\[
P_{\mathrm{succ}}^{\mathrm{qu}} = \sum_{k=1}^{d} \lambda'_{(k)}.
\]
\end{proof}

\begin{remark}
Since $\sum_{k=1}^{d}\lambda'_{(k)}$ is achieved by suitable ensembles, no bound depending only on the multiset of weighted eigenvalues and the dimension $d$ can be smaller. Because our measurement model is equivalent to arbitrary POVMs, \cref{thm:mixed-main} can be viewed as a general, spectrum-based upper bound on the optimal success probability of minimum-error discrimination for any finite mixed-state ensemble under arbitrary POVMs.
\end{remark}
\section{Conclusion}\label{sec:conclusion}
We studied a basic question: in a one-shot messenger--receiver problem without preshared entanglement, does sending a single qudit offer any advantage over sending a classical dit for identifying a random message drawn from a known prior distribution?

Using a simple first-principles measurement model, we showed that for pure-state encodings on a $d$-dimensional system,
\[
P_{\mathrm{succ}}^{\mathrm{qu}}\ \le\ \sum_{k=1}^{d}p_{(k)},
\]
matching the classical $d$-signal ceiling and demonstrating that qudits offer no advantage over dits in this setting. We then derived a mixed-state generalization,
\[
P_{\mathrm{succ}}^{\mathrm{qu}}\ \le\ \sum_{k=1}^{d}\lambda'_{(k)},
\]
where $\lambda'_{(k)}$ are the largest weighted eigenvalues of the states $\{p_i\rho_i\}$. Both bounds strictly improve on the standard dimension-only ceiling $P_{\mathrm{succ}}^{\mathrm{qu}}\le d\,p_{\max}$ and are sharp given only the priors and dimension, or priors plus spectra and dimension, respectively.

For the messenger--receiver problem itself, convexity implies that optimal encodings can be taken pure, so the qudit--vs.--dit statement is already captured by the pure-state result. The mixed-state spectral bound nonetheless provides a clean and general upper bound on the performance of arbitrary POVMs for minimum-error discrimination of mixed ensembles.

Two natural directions for further work stand out. First, one can ask what happens when \emph{preshared entanglement} between Alice and Bob is allowed. Superdense coding already shows that entanglement can increase the amount of classical information conveyed per transmitted qubit; it would be interesting to formulate and prove analogues of our sharp one-shot bounds in that entanglement-assisted setting for arbitrary prior distributions. Second, the present bounds involve the discrete dimension $d$ of the joint support. It would be useful to develop more ``continuous'' upper bounds that interpolate smoothly with the spectral profiles of the ensemble, for example by relating $P_{\mathrm{succ}}^{\mathrm{qu}}$ to suitable notions of effective rank or eigenvalue concentration rather than only the hard cutoff at $d$.

% ========================== Bibliography ========================== %

\nocite{*} % Remove this if you only want cited entries to appear
\printbibliography

\end{document}